\documentclass{amsart}

\usepackage{bbm,amsmath,amsfonts}
\usepackage{graphicx}
\usepackage{verbatim}
\usepackage[latin1]{inputenc}

\newtheorem{proposition}{Proposition}
\newtheorem{theorem}{Theorem}
\newtheorem{lemma}{Lemma}
\newtheorem{definition}{Definition}
\newtheorem{corollary}{Corollary}

\newcommand\ind[1]{\mathbbm{1}_{\{#1\}}}

\def\cal{\mathcal}

\def\Q{{\mathbb Q}}
\def\N{{\mathbb N}}
\def\R{{\mathbb R}}

\def\P{{\mathbb P}}
\def\E{{\mathbb E}}
\def\Var{\mathrm{Var}}

\def\etal{{\em et al.}}
\def\cvj{\xrightarrow[j\rightarrow +\infty]{}}

\title[An identification problem in an urn and ball model]{An identification problem in an
  urn and ball model with heavy tailed distributions} 
\date{}

\author{Christine Fricker}
\address[C. Fricker]{INRIA-Rocquencourt,  RAP project, Domaine de Voluceau, 78153 Le Chesnay, France}
\email{Christine.Fricker@inria.fr}
\urladdr{http://www-c.inria.fr/twiki/bin/view/RAP/ChristineFricker}
\author{Fabrice  Guillemin}
\address[F.~Guillemin]{Orange Labs, F-22300 Lannion}
\email{Fabrice.Guillemin@orange-ftgroup.com}
\author{Philippe Robert}
\address[Ph.~Robert]{INRIA-Rocquencourt,  RAP project, Domaine de Voluceau, 78153 Le Chesnay, France}
\email{Philippe.Robert@inria.fr}
\urladdr{http://www-rocq.inria.fr/\~{}robert}

\begin{document}

\begin{abstract}
We consider in this  paper an urn and ball problem with  replacement, where balls are with
different colors and  are drawn uniformly from a  unique urn. The numbers of  balls with a
given color are i.i.d. random variables  with a heavy tailed probability distribution, for
instance a  Pareto or a  Weibull distribution. We  draw a small  fraction $p\ll 1$  of the
total number  of balls.  The  basic problem addressed  in this paper  is to know  to which
extent we can infer the total number of colors and the distribution of the number of balls
with a given color. By means of  Le Cam's inequality and Chen-Stein method, bounds for the
total variation norm  between the distribution of  the number of balls drawn  with a given
color and the Poisson distribution with the  same mean are obtained. We then show that the
distribution of the number of balls drawn with  a given color has the same tail as that of
the  original number  of  balls. We  finally  establish explicit  bounds  between the  two
distributions when each ball is drawn with fixed probability $p$.
\end{abstract}

\keywords{Chen-Stein method, Pareto distribution, Weibull distribution.}

\maketitle

\section{Introduction}
We consider in this paper the following urn and ball scheme \emph{with replacement}: An
urn contains a random number of balls with different colors. We draw a small fraction
$p\ll 1$ of the total number of balls. A ball which has been drawn is 
replaced into the urn. The problem considered in this paper consists of estimating the
number of colors together with  the distribution of the number of balls with a given color by using
information from  sampled balls. This problem is motivated by the analysis of packet
sampling in the Internet (see Chabchoub \etal~\cite{Chabchoub:01} for details). 

To address the above problem, we
analyze the non-normalized distribution of the number of balls drawn with a given
color. More specifically, let $W_j$ (respectively,  $W_j^+$) denote the number of colors
with a number of sampled balls equal to (respectively, equal to or greater than)
$j$. Denoting by $\tilde{K}$ the number of colors seen when drawing balls, the quantities
$W_j/\tilde{K}$ and $W_j^+/\tilde{K}$ are equal to the proportions of colors, which at the
end of the trial comprise exactly or at least $j$ balls, respectively.  

The numbers of balls with various colors are assumed to be i.i.d. random variables and the
number $K$ of colors is  large. In addition, the distribution of the number of balls  with a given
color  has a heavy  tailed probability distribution  of Pareto or Weibull type. Finally, balls are
drawn uniformly. This means that for each $i=1,\ldots, K$, if there  are $v_i$ balls with
color  $i$,  the probability  of  drawing  a  ball with  this  color  is  $v_i/V$,  where
$V=v_1+\cdots+V_K$ is the total number of balls in the urn.  

The above model is defined as the ``uniform model''. It will  compared  against  the case when
balls are drawn independently one of each other with probability $p$. This model will be 
referred to as probabilistic model. We show that the results obtained in both models are
close one to each other when $p$ is very small. But there are some subtle differences
between the two models, notably with regard to the achievable accuracy in the inference of
original statistics. It turns out that the probabilistic model  is simpler to analyze than
the uniform model but yields less accurate results. This is due to the fact that we
cannot exploit the fact that the number of colors is very large. 

One of the main results of the paper concerns the analysis of the validity of the following
simple scaling  rule: The  distribution of the  original number
$v_i$  of balls  with color  $i$ could be estimated by  that of  the random  variable  $\tilde{v}_i/p$, where
$\tilde{v}_i$ is the number of sampled  balls with color $i$.  When each ball is drawn with a
fixed probability, it is  known that this rule is valid for  tails of the distributions as
soon as  they are heavy tailed. See  Asmussen \emph{et al}~\cite{Asmussen-2}  and Foss and
Korshunov~\cite{Foss-12}  where this asymptotic equivalence  is proved in  a quite general
framework. Our main goal  here  is to get,  for $j\geq 2$,  an {\em  explicit bound} on the quantity
\[
\left| \frac{\P(\tilde{v}\geq j)}{\P(v\geq j/p)}-1\right|.
\]
In the context of packet sampling in the Internet, explicit expressions are especially important for  the estimation of the sizes of flows in Internet traffic. In this setting
the variable $j$ is taken to be large but cannot be too large so that the event $\{\tilde{v} =j\}$ occurs
sufficiently often to obtain  reliable statistics. Henceforth, the dependence on $j$ should be
made explicit. See Chabchoub \etal~\cite{Chabchoub:01} for a discussion. 

The organization of this paper is as follows: The notation and the basic results used in
this paper (Le Cam's inequality and Chen-Stein method) are presented in
Section~\ref{notation}. The mean values of the random variables $W_j$ and $W_j^+$ are
computed in Section~\ref{means}. The approximation of the distribution of $W_j^+$ by a
Poisson distribution and the validity of the scaling rule are investigated in
Section~\ref{distrib}. We compare  in Section~\ref{comparison} the original distribution
of the number of balls  with a given color against the rescaled distribution of the number
of drawn balls with the same color. Some concluding remarks with regard to sampling are
presented in Section~\ref{conclusion}.

\section{Notation and basic results}
\label{notation}
\subsection{Definitions and assumptions}
We consider an urn containing $v_i$ balls with color $i$ for $i=1,\dots, K$. The quantities $v_i$ are independent random variables with a common heavy tailed distribution. In the following we shall consider two families of heavy tailed distributions for the number $v$ of balls with a given color:
\begin{description}
\item[Pareto distributions] The distribution of $v$ is given by
\begin{equation}
\label{pareto}
\P(v >  x) = (b/x)^a, \quad x \geq b,
\end{equation}
with the shape parameter $a >1$ and the location parameter $b >0$. The mean of $v$ is $ab/(a-1)$.
\item[Weibull distributions] The distribution of $v_i$ is given by
\begin{equation}
\label{weibull}
\P(v >  x) = \exp(-(x/\eta)^\beta), \quad x \geq 0,
\end{equation}
with the skew  parameter $\beta\in (0,1)$ and the scale parameter $\eta>0$. The mean of $v$ is $\frac{\eta}{\beta}\Gamma(1/\beta)$, where $\Gamma$ is the classical Euler's Gamma function.
\end{description} 

The total number of balls in the urn is $V= \sum_{i=1}^K v_i$. We draw only a fraction $p$ of this total number of balls. Each ball is drawn at random: A ball with color $i$ is drawn  with probability  $v_i/V$. After drawing the $pV$ balls, we have $\tilde{v}_i$ balls with color $i$. Of course, only those colors with $\tilde{v}_i>0$ can be seen. The quantity $\tilde{K} = \sum_{i=1}^K \ind{\tilde{v}_i>0}$ is the number of colors seen at the end of a trial.

In the following, we shall be interested in the asymptotic regime when the number of  colors $K \to \infty$ while the fraction  $p\to 0$. Note that by the law of large numbers,  $V \to \infty $ a.s. (the total number of balls in the urn is very large).   

The random variables we consider in this paper to infer the original statistics of the number of balls and colors are the 
variables $W_j$ and $W_j^+$, $j  \geq 0$,  defined as follows.

\begin{definition}[Definition of $W_j$] The random variable  $W_j$ is  the number of colors with  $j$ balls at the end of a trial and is given by
$$
j \geq 0, \quad W_j=\ind{\tilde{v}_1=j}+\ind{\tilde{v}_2=j}+\cdots+\ind{\tilde{v}_K=j},
$$
where $\tilde{v}_i \geq 0$ is the number of balls drawn with color $i$ (which can be equal to 0). 
\end{definition}

\begin{definition}[Definition of $W_j^+$] The random variable  $W_j^+$ is  the number of colors with at least $j$ balls at the end of a trial.   The random variables $W_j^+$ are formally defined by
$$
j \geq 0, \quad W^+_j=\ind{\tilde{v}_1\geq j}+\ind{\tilde{v}_2\geq j}+\cdots+\ind{\tilde{v}_K\geq j}.
$$
\end{definition}

Note that we have
$$
\forall j \geq 0, \quad W_j^+ = \sum_{\ell \geq j} W_\ell.
$$

The averages of the random variables $W_j$ are in fact the key quantities we shall use in the following to infer the original numbers of balls per color.

\subsection{Le Cam's inequality and Chen-Stein method}

Le Cam's inequality gives the distance in total variation between the distribution of a sum of independent and identically distributed (i.i.d.) Bernoulli random variables and the Poisson distribution with the same mean (see Barbour \etal~\cite{Barbour}). Note that if $V$ and $W$ are two  random variables taking integer values, the distance in total variation between their distributions is defined by
\begin{eqnarray*}
\|\P(W\in \cdot)-\P(V \in\cdot)\|_{tv}  &\stackrel{\text{def.}}{=}&   \sup_{A\subset \N} \left|\P(W\in A)-\P(V\in A)\right| \\
&=& \frac{1}{2}\sum_{n\geq 0} \left|\P(W=n)-\P(V=n)\right|.
\end{eqnarray*}

\begin{theorem}[Le Cam's Inequality]
If the random variable $W=\sum_{i}I_i$, where the random variables $I_i$ are i.i.d. Bernoulli random variables, then
\begin{equation}\label{lecam}
\|\P(W\in \cdot)-\P(Q_{\E(W)}\in\cdot)\|_{tv}\leq \sum_{i} \P(I_i=1)^2,
\end{equation}
where for $\lambda >0$, $Q_\lambda$ is a Poisson random variable with mean $\lambda$, that is, for all $n \geq 0$, $$\P(Q_{\lambda} = n) = \frac{\lambda^n}{n!} e^{-\lambda}.$$
\end{theorem}

When the random variables $I_i$ appearing in the above theorem are not independent but satisfy a specific condition, referred to as monotonic coupling, it is still possible to obtain a bound on the distance between the distribution of the sum $W = \sum_i I_i$ and the Poisson distribution with mean $\E(W)$. 

\begin{definition}[Monotonic Coupling]
The variables $I_i$ are said to be {\em negatively related}, when there exist some random
variables $U_i$ and $V_i$ such that
\begin{enumerate}
\item $U_i\stackrel{\text{dist.}}{=} W$ and
  $1+V_i\stackrel{\text{dist.}}{=} (W\mid I_i=1)$; 
\item $V_i\leq U_i$.
\end{enumerate}
\end{definition}

The main result of the Chen-Stein method is given by  the following theorem (see
Barbour \etal~\cite{Barbour}). 

\begin{theorem}\label{theobar}
If {\em the monotonic coupling condition} is satisfied, then the following inequality holds
\begin{equation}\label{varE}
\|\P(W\in \cdot)-\P(Q_{\E(W)}\in\cdot)\|_{tv}\leq 1-\frac{\Var(W)}{\E(W)}.
\end{equation}
\end{theorem}

When  the  monotonic coupling  condition  is  satisfied, in  order  to  prove the  Poisson
approximation, it is sufficient to show that the ratio of the variance to the mean
value of $W$ is close to $1$; this is a very weak condition to prove in practice.

It  should be noted  (see \cite{Robert-5}) that Relation~\eqref{varE}  can be  used not  only when  $\E(W)$ takes
bounded values  so that $W$ is  approximately a  Poisson random variable, but  also when
$\E(W)$ is large. In this case Chen-Stein Method yields a central limit theorem: If ${\cal
N}$ is a standard normal distribution,
\begin{multline*}
\left\|\P\left(\frac{W-\E(W)}{\sqrt{\Var(W)}}\in\cdot \right)-\P({\cal N}\in\cdot)\right\|_{tv}
\leq\\
 \left\|\P\left(\frac{W{-}\E(W)}{\sqrt{\Var(W)}}\in\cdot \right){-}\P\left(\frac{Q_{\E(W)}{-}\E(W)}{\sqrt{\Var(W)}}\in\cdot \right)\right\|_{tv}
\\{+}\left\|\P\left(\frac{Q_{\E(W)}{-}\E(W)}{\sqrt{\Var(W)}}\in\cdot \right){-}\P({\cal N}\in\cdot)\right\|_{tv}{}
\end{multline*}
where $\Var(W)$ is the variance of the random variable $W$.

By using Relation~\eqref{varE}, we have
\begin{multline*}
\left\|\P\left(\frac{W-\E(W)}{\sqrt{\Var(W)}}\in\cdot \right)-\P({\cal N}\in\cdot)\right\|_{tv}
\leq 1-\frac{\Var(W)}{\E(W)}
\\+\left\|\P\left(\frac{Q_{\E(W)}{-}\E(W)}{\sqrt{\Var(W)}}\in\cdot \right){-}\P({\cal N}\in\cdot)\right\|_{tv}.
\end{multline*}
If the ratio $\E(W)/\Var(W)$  is close to $1$, then the first term  in the right hand side
of the above relation is negligible.  In addition, the classical central limit theorem for
Poisson  distributions   implies  that  when  $\E(W)$   is  large,  the   second  term  is
negligible too. Therefore, we  have $W\sim \E(W)+\sqrt{\Var(W)} {\cal N}$ with  a bound on the
error.

\section{Computation of mean values}
\label{means}
\subsection{Bounds for mean values}

By using Le Cam's inequality, we can establish the following result for the mean value of the random variables $W_j$. 

\begin{proposition}[Mean Value of $W_j$]\label{boundobs}
If there are $V$ balls and   $K$ colors in the urn, for $j \geq 0$, the mean number $\E(W_j)$ of colors with $j$ balls at the end of a trial  satisfies the relation 
\begin{align}\label{asymp}
\left|\frac{\E(W_j)}{K}-\Q_{j}\right|\leq  \E\left(\min(pv,1) \frac{v}{V}\right),
\end{align}
where $\Q$  is the probability distribution defined for $ j \geq 0$ by 
\begin{align*}
 \Q_{j}=\E\left(\frac{{\left(p v\right)}^j}{j!}e^{-{p}v}\right),
\end{align*}
$p$ is the sampling rate, and $v$ is distributed as the number of balls with a given color.
\end{proposition}

\begin{proof}
We have
\[
\tilde{v}_i=  B_{1}^i+B_{2}^i+\cdots+B_{p V}^i,
\]
where $B_\ell^i$ is equal to one if the  $\ell$th ball drawn from the urn has color $i$, which event occurs with probability $v_i/V$, the quantity $V$ being the total number of balls in the urn. 

Conditionally on the values of the set ${\cal F}=\{v_1, \ldots, v_K\}$, the variables $(B_\ell^i, \ell \geq 1)$ are
independent Bernoulli variables. 
For $1\leq i\leq K$, Le Cam's Inequality~\eqref{lecam} therefore gives  the
relation 
\[
\left\|\P(\tilde{v}_i\in\cdot\mid {\cal F})-\P(Q_{p {v}_i}\in \cdot)  \right\|_{tv}
\leq p\frac{v_i^2}{V},
\]
and Relation~\eqref{varE} which can also be used in this case yields
\[
\left\|\P(\tilde{v}_i\in\cdot\mid {\cal F})-\P(Q_{p {v}_i}\in \cdot)  \right\|_{tv}
\leq \frac{v_i}{V},
\]
By integrating with respect to the variables $v_1, \ldots, v_K$,  these two inequalities
give  the relation
\begin{equation}
\label{distWj}
\left\|\P(\tilde{v}_i\in\cdot)-\Q\right\|_{tv} \leq
\E\left(\min\left(pv,1\right)\frac{v}{V}\right).
\end{equation}
Since $\E(W_j)=\sum_{i=1}^K\P(\tilde{v}_i=j)$,
by summing on $i=1, \ldots, K$,  we obtain
\[
\left|\E(W_j)-K\Q_j\right|\leq  K \E\left(\min\left(pv,1\right)\frac{v}{V}\right).
\]
and the result follows.
\end{proof}

By using the fact that $\E(W_j^+) = \sum_{i=1}^K \P(\tilde{v}_i \geq j)$, we can deduce
from Equation~\eqref{distWj} the following result. 

\begin{proposition}[Mean Value of $W_j^+$]\label{boundobs+}
If there are $V$ balls and  $K$ colors in the urn, the mean number $\E(W^+_j)$ of colors with at least $j \geq 0$ balls at the end of an arbitrary trial  satisfies the relation 
\begin{align}\label{asymp+}
\left|\frac{\E(W^+_j)}{K}-\sum_{\ell \geq j} \Q_{\ell}\right|\leq \E\left(\min\left(pv,1\right)\frac{v}{V}\right),
\end{align}
where the probability distribution $\Q$ is defined in Proposition~\ref{boundobs}.
\end{proposition}

We immediately deduce from Propositions~\ref{boundobs} and \ref{boundobs+} the following
corollary by using the fact that $V\geq K$. 

\begin{corollary}[Asymptotic Mean Values] \label{corasymp}
The relations 
$$
\lim_{K\to \infty}\frac{1}{K}\E(W_j)= \Q_j \quad \mbox{and} \quad \lim_{K\to \infty}\frac{1}{K}\E(W^+_j)= \sum_{\ell \geq j}\Q_\ell .
$$
hold.
\end{corollary}

Note that if balls are  drawn with probability $p$ independently one of each other (probabilistic model), we have $\tilde{v}_i = \sum_{\ell=1}^{v_i} \tilde{B}^i_\ell$, where the random variables $\tilde{B}^i_\ell$ are Bernoulli with mean $p$. By adapting the above proofs, we find
\begin{equation}
\label{Wjproba}
\left|\frac{\E(W_j)}{K}- \Q_{j}\right|\leq p.
\end{equation}

\subsection{Asymptotic results for specific probability distributions}

\subsubsection{Pareto distributions}
Let us first assume that the number of balls of a given color follows a Pareto distribution given by Equation~\eqref{pareto}. Then, we have the following result when the number of colors goes to infinity.

\begin{proposition}
\label{propalpareto}
If $v$ has a Pareto distribution as in Equation~\eqref{pareto}, then for all
$j> a$, the relations
\begin{align}
\lim_{K\to +\infty} \frac{\E(W_{j+1})}{\E(W_{j})} & =  1-\frac{a+1}{j+1}+O((pb)^{j-a}), \label{equiv} \\
\lim_{K\to +\infty} \frac{\E(W_j)}{K}& = a(pb)^{a} \frac{\Gamma(j-a)}{j!}+O((pb)^{j}), \label{Kel} \\
 \lim_{K\to +\infty} \frac{\E(W^+_j)}{K} &= (p b)^{a} \frac{\Gamma(j-a)}{(j-1)!}+O\left(\frac{(pb)^j}{1-pb}\right) \label{Kelbis}
\end{align}hold.
\end{proposition}

\begin{proof}
For $j>a$,
\begin{multline}\label{equivQj}
\Q_{j}=\E\left(\frac{{\left(p v\right)}^j}{j!}e^{-{p}v}\right) =  ab^a \frac{p^{a} }{j!} \int_{p b}^{+\infty} u^{j-a-1} e^{-u}\,du \\=
a(p b )^{a}\frac{\Gamma(j-a)}{j!}-a\frac{(pb)^{j}}{j!}\int_0^1 u^{j-a-1}e^{-pb u}\,du.
\end{multline}
Therefore, by using the relation $\Gamma(x+1)=x\Gamma(x)$, we get the equivalence
\[
\frac{\Q_{j+1}}{\Q_j}=  \frac{j-a}{j+1}+O((pb)^{j-a}),
\]
which gives Equations~\eqref{equiv} and \eqref{Kel} by using Corollary~\ref{corasymp}.  
For the  mean value of $W_j^+$, Equation~\eqref{equivQj} gives the relation
\begin{align*}
\lim_{K\to +\infty} \frac{\E(W_j^+)}{K}  
&=  a(p b)^{a} \sum_{n\geq j}\frac{\Gamma(n-a)}{n!} +O\left(\frac{(pb)^j}{1-pb}\right)\\
&=  a(p b)^{a} \sum_{n\geq  0}\frac{\Gamma(n+j-a)\Gamma(n+1)}{\Gamma(j+n+1)} \frac{1^n}{n!}+O\left(\frac{(pb)^j}{1-pb}\right)\\
&=  a(p b)^{a} \frac{\Gamma(j-a)}{j!}F(j-a,1;j+1;1)+O\left(\frac{(pb)^j}{1-pb}\right),
\end{align*}
where $F(a,b;c;z)$ is the hypergeometric function satisfying
\[
F(a,b;c;1)=\frac{\Gamma(c)\Gamma(c-a-b)}{\Gamma(c-a)\Gamma(c-b)}
\]
(see Abramowitz and Stegun~\cite{Abramowitz}), and 
Equation~\eqref{Kelbis} follows.
\end{proof}

The shape parameter $a$ can be estimated via Relation~\eqref{Kelbis} by
\begin{equation}
\label{estima}
a = \lim_{K\to \infty} j \left(1- \frac{\E(W^+_{j+1})}{\E(W^+_j)}  \right) +O\left(\frac{(pb)^j}{1-pb}\right)
\end{equation}
for all $j>a$. This gives a means of estimating the shape parameter $a$. When observing drawn balls, we
have in fact only access to the quantity $\E(\tilde{K})$ of the number of sampled colors.  While this has no impact for the estimation of $a$,  this correcting term
is important when estimating $b$ from Equation~\eqref{Kelbis}. It is straightforward that  
$$
\tilde{K} = \sum_{i=1}^K \ind{\tilde{v}_i>0} =K-W_0
$$
and then when $K \to \infty$
$$
\E(\tilde{K}) \sim K(1-\Q_0) = K\left( 1-\E(e^{-pv})\right).
$$  
Since 
\begin{equation}
\label{estimnu}
1-\E(e^{-pv}) =  p \int_0^\infty e^{-px}\P(v>x) dx = bp + (bp)^a \Gamma(1-a,bp),
\end{equation}
where $\Gamma(a,x)$ is the incomplete Gamma function defined by $\Gamma(a,x) = \int_x^\infty t^{a-1}e^{-t}dt$, we can use the above equations together with Equation~\eqref{Kelbis} in order to  estimate $b$ and then $K$. It is also worth noting that $ 1-\E(e^{-pv})    \sim bp$ when $a>1$ and $bp \to 0$. 


\subsubsection{Weibull distributions} We assume in this section that the number of balls with a given color follows a Weibull distribution. In this case, we  have the following result, which follows from a simple variable change and the expansion of  $\exp(-x^\beta)$ in power series of $x^\beta$ or $\exp(-px)$ in power series of $x$; the proof is omitted.

\begin{proposition}
If $v$ has a Weibull distribution with skew parameter $\beta$ and scale parameter $\eta$, then for $0<\beta<1$
\begin{equation}\label{equiw}
\lim_{K\to +\infty} \E(W_{j+1}) = \frac{\beta}{j!}  \sum_{n=0}^\infty \frac{(-1)^n}{(p\eta)^{(n+1)\beta}}\frac{\Gamma((n+1)\beta+j)}{n!}
\end{equation}
and for $\beta>1$,
\begin{equation}\label{equiw+}
\lim_{K\to +\infty} \E(W_{j+1}) = \frac{(p\eta)^j}{j!}  \sum_{n=0}^\infty \frac{(-p\eta)^n}{n!}\Gamma\left(\frac{(n+j)}{\beta}+1\right).
\end{equation}
\end{proposition}

Note that $\E(W_j)$ can be written in the form
$$
\E(W_j) = \frac{1}{j!} \frac{\beta}{(p\eta)^\beta} \int_0^\infty u^{j+\beta-1}e^{-u+tu^\beta}du
$$
with $t = -1/(p\eta)^\beta$. The above integral is known in the literature as to be of the
Faxen's type and can be expressed  by means of Meijer $G$-function, when $\beta$ is a
rational number, see Abramowitz and Stegun~\cite{Abramowitz}.

Contrary to the case of Pareto distribution for the initial distribution of balls of a given color, there is no simple relations giving the parameters $\beta$ and $\eta$ from the mean values $\E(W_j)$, $j \geq 1$. In fact, we shall prove in the following that $\P(\tilde{v}\geq j)$ has also a Weibull tail. This eventually gives a means of identifying the parameters.

\section{Poisson approximations}
\label{distrib}

In the previous section, we have established bounds for the mean values of the random
variables $W_j$ and $W_j^+$. To obtain more information on their distributions, we intend
to use Chen-Stein method. For a fixed environment (namely fixed values of the quantities
$v_i$ for $i=1, \ldots,K$), these random variables appear as sums of non independent
Bernoulli random variables. A preliminary analysis of the Bernoulli random variables
appearing in the expression of $W_j$ reveals that it seems not possible to invoke a
monotonic coupling argument. It is well known (see \cite{Barbour} for details) that the
situation is more favorable with the random variables $W^+_j$ and we can specifically
prove that  if ${\cal F}$ is the set ${\cal F}=\{v_i, 1\leq i\leq K\}$, then the total
number $W_j^+$ of colors with at least $j$ balls at the end of the trial satisfies the
relation 
\begin{equation}\label{aa}
\left\|\P(W_j^+{\in} \cdot ~|~ \mathcal{F}){-}\P(Q_{\E(W_j^+~|~\mathcal{F})}{\in}\cdot\,)\right\|_{tv}{\leq}
\E\left(1{-}\frac{\Var(W_j^+\mid {\cal F})}{\E(W_j^+\mid {\cal F})}\right).
\end{equation}
Indeed, given the random variables $v_i$, the model is equivalent to a standard urn and ball   problem consisting of putting  $pV_i$
balls  into $K$ urns, a  ball falling into urn $i$ with probability $p_i=v_i/V_i$. The
number of balls in urn $i$ is the number of balls with color $i$ in the original urn and ball problem. Even in  the case when the quantities $p_i$ are different, the
variables $I^+_{i,j}\stackrel{def}{=} \ind{\tilde{v}_i\geq j}$ are negatively
related so that Theorem~\ref{theobar} can be used. See Page~24 and Corolary~2.C.2 Page~26
of \cite{Barbour} for a definition and the main inequality in this domain. Chapter~6 of
this reference is entirely devoted to related occupancy problems.

The rest of this section is devoted to the estimation of the bound in Equation~\eqref{aa}. We first establish the following lemma.

\begin{lemma} \label{PA}
For a fixed environment ${\cal F}=\{v_i,1\leq i\leq K\}$, the distance in total variation between the distribution of $W^+_j$ and the  Poisson distribution
$Q_{\E(W^+_k~|~{\cal F})}$ satisfies the inequality
\begin{align}\label{bizerte}
\lim_{K\rightarrow +\infty} \|   \P(W^+_j  \in \cdot~|~ {\cal F})- \P(Q_{\E(W^+_k~|~{\cal  F})}\in\cdot)\| _{tv}\leq
\frac{m_{2,j}(p)}{m_{j}(p)}+\frac{p}{\E(v)} \frac{m'_{j}(p)^2}{m_{j}(p)},
\end{align}
where
$m_j(p)$ and $m_{2,j}(p)$ are the first two moments of the random variable defined by
\begin{equation}
\label{defXk}
X_j(p)=\sum_{\ell \geq j} \frac{(pv)^\ell}{\ell!} e^{-pv},
\end{equation}
and the prime sign denotes the derivative with respect to $p$.
\end{lemma}

\begin{proof}
For $\mathcal{F}$ fixed, the number $W_j$ of colors with $j \leq pV$ balls at the end of the trial is such that
\[
\E(W_j~|~{\cal F})=\sum_{i=1}^K \binom{pV}{j}\left(\frac{v_i}{V}\right)^j\left(1-\frac{v_i}{V}\right)^{pV-j}.
\]
By using the fact that 
$$
\frac{1}{V}  = \frac{1}{K\E(v)} + o\left(\frac{1}{K}\right) \quad \mbox{a.s.}
$$
for large $K$, straightforward calculations show that
\begin{multline}
\label{eqtech}
\E(W_j ~|~{\cal F}) = \sum_{i=1}^K \frac{(pv_i)^j}{j!} e^{-pv_i}\left(1-\frac{j(j-1)}{2pK\E(v)}+\frac{2jv_i-pv_i^2}{2\E(v)K}\right)+o\left(\frac{1}{K}\right)\\
= \sum_{i=1}^K \left( \frac{(pv_i)^j}{j!}e^{-pv_i}
-\frac{p}{2\E(v)K}\frac{d^2}{dp^2}\left(e^{-pv_i}\frac{(pv_i)^j}{j!}\right)\right)+o\left(\frac{1}{K}\right).
\end{multline}
By summing up the terms above and by checking that the
$o\left(\frac{1}{K}\right)$ term remains valid, since the sum can be written as  $\sum_{i=1}^K f(v_i)e^{-pv_i}/K^2$, where $f$ is a polynomial, we have for $j\geq 1$ and $0<p<1$
\[
\E(W^+_j~|~{\cal F})=\sum_{\ell \geq j} \E(W_\ell ~|~{\cal   F})=\sum_{i=1}^K X_{i,j}(p)-\frac{p}{2\E(v)K}  \sum_{i=1}^K X''_{i,j}(p)+o\left(\frac{1}{K}\right),
\]
where
$$
X_{i,j} (x) = \sum_{\ell \geq j} \frac{(xv_i)^\ell}{\ell!}e^{-x v_i}.
$$

For the variance, if $I_{i,j}$ is 1 if color  $i$ has exactly $j$ balls at the end of the trial and  0 otherwise, then  $W_j =\sum_{i=1}^K I_{i,j}$ and, for $j \neq \ell$,
$$
\E(W_j W_\ell ~|~ {\cal F})=\sum_{1\leq i\not=m\leq K}\E(I_{i,j}I_{m,\ell}~|~{\cal F})
$$
and
$$
\E(W_j^2~|~{\cal F})=\E(W_j~|~{\cal F})+\sum_{1\leq i\not=m\leq K}\E(I_{i,j}I_{m,j}~|~{\cal F}).
$$
For $j,\ell$ such that $j+\ell \leq pV$, 
$$
\E(I_{i,j}I_{m,\ell}~|~{\cal F}) = 
\frac{(pV)!}{j!\ell!(pV-j-\ell)!}\left(\frac{v_i}{V}\right)^j\left(\frac{v_m}{V}\right)^\ell\left(1-\frac{v_i+v_m}{V}\right)^{pV-j-\ell}.
$$
The quantity in the right hand side of the above equation can be expanded  as
$$
\frac{e^{-p(v_i+v_m)}p^{j+\ell}v_i^j v_m^\ell}{j!\ell!} - \frac{p}{2V} \frac{e^{-p(v_i+v_m)} v_i^j v_m^\ell}{j!\ell!}c_{i,m}(j,\ell)   +o\left(\frac{1}{K}\right),
$$
where
$$
c_{i,m}(j,\ell) =
p^{j+\ell -2}(j+\ell)(j+\ell -1)-2(j+\ell)(v_i+v_m)p^{j+\ell-1}+(v_i+v_m)^2 p^{j+\ell} 
$$
is such that
$$
 \frac{e^{-p(v_i+v_m)}v_i^jv_m^\ell}{j!\ell!} c_{i,m}(j,\ell)  =   \frac{d^2}{dp^2}    \frac{e^{-p(v_i+v_m)} v_i^j v_m^\ell}{j!\ell!}.
$$
Since
\[
(W^{+}_j)^2 =\left(\sum_{\ell \geq j} W_\ell\right)^2=\sum_{\ell \not=k \geq
  j}W_kW_\ell +\sum_{\ell\geq j}W_\ell^2,
\]
\begin{multline*}
\E((W^{+}_{j})^2 ~|~{\cal F})- \E(W_j^+~|~{\cal F})= \sum_{1\leq i \not=m  \leq K} \sum_{\ell,k\geq j} \E(I_{i,k}I_{m,\ell}~|~{\cal F})\\
= \sum_{1\leq i\not=m\leq K} \left(X_{i,j}(p)X_{m,j}(p)-\frac{p}{2\E(v)K}    \left( X_{i,j}X_{m,j}\right)'' (p)\right) +o\left(\frac{1}{K}\right),
\end{multline*}
and
$$
1-\frac{\Var(W^+_j~|~{\cal F})}{\E(W^+_j~|~{\cal F})} = \frac{\E(W^+_j~|~{\cal F})-\E((W^{+}_{j})^2 ~|~{\cal F})+\E(W^+_j~|~{\cal F})^2}{\E(W^+_j~|~{\cal F})}.
$$
The right-hand side of this equation  can be expanded as
\begin{multline*}
\frac{1}{\sum_{i=1}^K X_{i,j}+O(1)}\left(  -\sum_{1 \leq i \not =m \leq K}
  X_{i,j}(p)X_{m,j}(p)\right. \\ 
\left. +\frac{p}{2\E(v)K}\sum_{1 \leq i \not =m \leq  K} (X_{i,j}X_{m,j})''(p) +\left(\sum_{i=1}^K X_{i,j}(p) -\frac{p}{2\E(v)K} \sum_{i=1}^K X_{i,j} ''(p)\right)^2\right)  \\ + o\left(\frac{1}{K}\right)
\end{multline*}
which  can be rewritten as
\begin{multline*}
\frac{1}{\sum_{i=1}^K X_{i,j}+O(1)}\left( \sum_{1 \leq i \leq  K} X_{i,j}^2(p) \right. \\ 
\left. + \frac{p}{2\E(v)K}\left(\sum_{1 \leq i \not =m \leq  K}(X_{i,j}X_{m,j})''(p)-2\sum_{i=1}^K X_{i,j}(p) \sum_{i=1}^K X''_{i,j}(p)\right) \right)+O(1)
\end{multline*}
  using that  
\[
\sum_{i\not= m} X_{i,j} X_{m,j}=\left(\sum_iX_{i,j}\right)^2-\sum_i X_{i,j}^2.
\] 
By the law of large numbers, we have that, almost surely, 
\begin{align*}
\lim_{K \to +\infty}\frac{1}{K} \sum_{i=1}^K X_{i,j}^2 (p)= \E(X_j^2 (p))= m_{2,j}(p),\\
\lim_{K \to +\infty} \frac{1}{K^2}\sum_{i\not = m}^K (X_{i,j}X_{m,j})'' (p)=(m_{j}^2)''(p),
\end{align*}
together with 
\[
\lim_{K \to +\infty} \frac{1}{K} \sum_{i=1} X_{i,j}(p)= m_{j}(p) \quad \mbox{and} 
\quad \lim_{K \to +\infty} \frac{1}{K} \sum_{i=1}^K     X{''}_{i,j}(p) =   m_{j}''(p).
\]
Hence, 
\begin{align*}
\lim_{K\to \infty} 1-\frac{\Var(W^+_j~|~{\cal F})}{\E(W^+_j~|~{\cal F})} &=
\frac{m_{2,j}(p)+p[(m_{j}^2)''(p)/2-m_{j}(p)m_{j}''(p)]/\E(v)}{m_{j}(p)}\;\;  a.s.\\
&= \frac{m_{2,j}(p)+{p}m_{j}'(p)^2/\E(v)}{m_j(p)} \quad a.s.
\end{align*}
and the result follows.
\end{proof}

To illustrate the fact that the bound in Equation~\eqref{bizerte} is tight when $p\to 0$
and $v$ has finite moments of any order, let us note that, provided the corresponding
moments are finite, 
\begin{equation}\label{orme}
\lim_{p\to 0}  \frac{m_j(p)}{p^j}= \frac{v^j}{j!}
\end{equation}
Moreover,
\[
\lim_{p\to 0} \frac{m_{2,j}(p)}{p^{2j}}= \frac{\E(v^{2j})}{j!^2}  \quad \mbox{and}
\quad \lim_{p\to 0} \frac{m'_j(p)}{p^{j-1}}=\frac{\E(v^j)}{(j-1)!}.
\]
Thus, the limit when $K$ tends to $+\infty$ of the bound given by Equation~\eqref{bizerte} is equivalent  to
\[
\frac{j p^{j-1}}{(j-1)!} \frac{\E(v^j)}{\E(v)}
\]
when $p$ tends to 0. If $j \geq 2$, this term tends to 0 when $p\to 0$.

By using the above lemma, we are now able to state a  limit result for the
distribution of the random variables $W_j^+$. 
\begin{proposition}\label{cltWj+}
The inequality
\begin{multline}\label{bound}
\lim_{K\to +\infty}
\sup_{y\in\R}\left|\P\left(\frac{W_j^+-\E(W_j^+)}{\sqrt{\E(W_j^+)}}\leq y\right)-
\int_{-\infty}^y \frac{e^{-u^2/2}}{\sqrt{2\pi}}\,du\right| \\ 
\leq  \frac{m_{2,j}(p)}{m_j(p)}+\frac{p}{\E(v)}\frac{(m_{j}'(p))^2}{m_j(p)}
\end{multline}
holds. 
\end{proposition}
\noindent
Thus, for $j\geq 2$ and for small $p$, this gives the following approximation
\[
W_j^+\sim  \E(W_j^+)+\sqrt{ \E(W_j^+)}, 
\]
where $G$ is a standard normal random variable. It should be noted nevertheless that
Equation~\eqref{bound} is almost a central limit result but because of the scaling in 
$1/\sqrt{\E(W_j^+)}$ instead of $1/\sqrt{\Var(W_j^+)}$, the bound in the right hand side
is not $0$ as $K$ gets large but, according to the proof of Lemma~\ref{bizerte}, only an
upper bound on the distance between $\E(W_j^+)$ and $\Var(W_j^+)$.

\begin{proof}
From Lemma~\ref{PA}, we have
\begin{multline*}
\left\| \P\left(\frac{W_j^+ - \E(W_j^+)}{\sqrt{\E(W_j^+)}} \in \cdot ~|~\mathcal{F}   \right) - \P\left(\frac{Q_{\E(W_j^+~|\mathcal{F})}  - \E(W_j^+|\mathcal{F})}{\sqrt{\E(W_j^+|\mathcal{F})}} \in \cdot  \right) \right\|_{tv} \\ \leq  \frac{m_{2,j}(p)}{m_{j}(p)}+\frac{p}{\E(v)} \frac{m'_{j}(p)^2}{m_{j}(p)}.
\end{multline*}
From Equation~\eqref{eqtech}, we have that 
\[
\lim_{K \to \infty}\frac{1}{K}\E(W_j^+~|~\mathcal{F})
=\E(X_j(p)) = K\sum_{\ell \geq j} \Q_\ell=Km_j(p),
\]
where the quantities $\Q_\ell$
are defined in Proposition~\ref{boundobs}. In addition, from  Corollary~\ref{corasymp},
$\E(W_j^+)\sim Km_j(p) $ when $K\to +\infty$. The result then follows by applying the
central limit theorem for Poisson distributions and by deconditioning with respect to
$\mathcal{F}$. 
\end{proof}

To conclude this section, let us notice that when balls are drawn with probability $p$
independently of each other, we do not have to condition on the environment and we
have 
$$
\left\|\P(W_j^+{\in} \cdot ){-} \P(Q_{\E(W_j^+)}\in \cdot) \right\|_{tv}{\leq} \frac{\E\left(\sum_{k=j}^v \binom{v}{k}p^k (1-p)^{v-k}\ind{v\geq j}  \right)^2}{ \E\left(\binom{v}{j}p^j(1-p)^{v-j} \ind{v\geq j}\right)},
$$
It is worth noting that the results are independent of the number of colors and that we do not need take $K \to \infty$ to obtain a bound for the distance in total variation. In addition, when $\E(W_j)$ become large, then it is possible to obtain a central limit-type approximation  similar to Proposition~\ref{cltWj+}.

\section{Comparison with original distributions}
\label{comparison}
\subsection{Uniform model}

In this section, we compare the distribution of the number $\tilde{v}$ of balls drawn with a given color with that of the original  number $v$ of balls with a given color. We are in particular interested in giving a sense to the heuristic stating that $v$ and $\tilde{v}/p$ have distributions close to each other.

\begin{proposition}\label{taildist}
Under the condition that the random variable $v$ has a Weibull or Pareto distribution, we have
$$
\lim_{j\to \infty}\lim_{K\to \infty} \frac{\E(W^+_j)}{K\P(v\geq j/p)}=1.
$$
\end{proposition}

\begin{proof}
From Corollary~\ref{corasymp}, we know that $\E(W_j)/K \to \Q_j$ when $K \to \infty$. Since 
$$
\Q_j = \E\left(\frac{(pv)^j}{j!} e^{-pv} \right) = \sum_{\ell = 1}^\infty \frac{(p\ell)^j}{j!} e^{-p\ell} \P(v=l),
$$
we can show that if $v$ has a Weibull or Pareto  distribution, then   $\Q_j \sim \P(v=j/p)/p$ when $j \to \infty$. Indeed, the above sum can be rewritten as
$$
 \frac{1}{j!} \sum_{\ell = 1}^\infty  e^{f_j(\ell)} \P(v=\ell),
$$
where $f_j(\ell) = -p\ell+j\log(p\ell)$, which attains its maximum at point $j/p$ with
$f''_j(j/p) = -p^2/j$. If the random variable $v$  is Weibull or Pareto and $j/p$ is
sufficiently large, then $\P(v=\ell)/\P(v=j/p) - 1\sim 0$ uniformly  on $j$ for $\ell$ in the neighborhood of
$j/p$. It follows that 
$$
\Q_j \sim \frac{1}{j!} \P(v=j/p)e^{f_j(j/p)} \sum_{\ell=-\infty}^\infty e^{-\ell^2 \frac{p^2}{2j}}.
$$
For $a>0$ converging to $0$,
\begin{multline*}
\sum_{\ell=-\infty}^\infty e^{-a\ell^2} =\sum_{\ell=-\infty}^\infty\int_{0}^{+\infty}
\ind{u>a\ell^2} e^{-u}\,du \sim 2 \int_0^{+\infty} \sqrt{\frac{u}{a}} e^{-u}\,du
\\=2 \int_0^{+\infty} \frac{u^2}{\sqrt{a}} e^{-u^2/2}\,du=\sqrt{\frac{\pi}{a}}
\end{multline*}
and by Stirling formula $j! \sim \sqrt{2\pi}j^{j+\frac{1}{2}}e^{-j}$ for large $j$, so that  $\Q_j \sim \P(v=j/p)/p$.  It is then easy to deduce  that $\sum_{\ell \geq j} \Q_j \sim \P(v\geq j/p)$ for large $j$. 
\end{proof}

The above Proposition implies that $\P(\tilde{v}\geq j)$ is such that $\P(\tilde{v}\geq j)\sim \P(v \geq j/p)$ when the number of colors is large. This means that the tail of the  distribution of the random variable ${v}$ can be obtained by rescaling that of the  number $\tilde{v}$ of sampled balls with a given color. When $v$ has a Pareto distribution, Equation~\eqref{estima} can still be used for large $j$ to estimate the shape parameter $a$. The estimation of the probability $1-\E(e^{-p v})$ of sampling a color and the scale parameter $b$ can also be estimated from the tail by using the expression of that probability as a function of $b$ and $a$ as in Equation~\eqref{estimnu}. The same method applies for Weibull distributions.

\subsection{Probabilistic model}

From now on, we consider the probabilistic model and we establish stronger results on the distance between $\P(\tilde{v}\geq j)$ and $\P(v\geq j/p)$, where $\tilde{v}$ is the number of balls with a given color at the end of a trial.  For this sampling mode, it  was not possible to prove a result similar to Corollary~\ref{corasymp}, but Berry-Essen's theorem \cite{Feller} can be used to  establish a stronger result for the comparison between $\tilde{v}$ and $v$. In \cite{Chabchoub:01}, it is specifically proved that if  we define the function $h_j(x)={x^2}/{4p^2}\left(\sqrt{1+{4jp}/{x^2}}-1\right)^2$ for $x\in \R$ and $j>0$, then 
\[
\left|\P\left(\tilde{v} \geq j\right)-
\P\left(v \geq h_j\left(\sqrt{p(1-p)}\cal{G}\right)\vee k\right) \right| \leq c\E\left(\frac{1}{\sqrt{v}}\ind{v\geq j}\right),
\]
where $\cal{G}$ is a standard Gaussian random variable, for real numbers $a \vee b = \max(a,b)$, and 
$c=3(p^2+(1-p)^2)/\sqrt{p(1-p)}$. For small $p$, the constant $c \sim 3/\sqrt{p}$. The above bound is very loose for small $p$ and becomes accurate  only for very large values of $j$. This is why we go further in this paper by establishing  a tighter bound for the ratio ${\P(\tilde{v}\geq j)}/{\P(v\geq j/p)}$.

Let $(B_n)$ be some sequence of i.i.d. Bernoulli random variables with
parameter $p$ and $v$ some independent r.v. on $\N$. Take some
$\alpha\in ]1/2,1[$. Let $\tilde{v}=\sum_{l=1}^v B_l$.

 \begin{theorem}
 \label{distance2}
 For $\alpha\in(1/2,1)$, we have for all $j \geq 1$
 $$
 \frac{\P(\tilde{v}\geq j)}{\P(v\geq j/p)} = A(j) + B(j),
 $$
 where
 $$
 A_1(j) \leq A(j) \leq A_2(j)
 $$
 with 
 \begin{align*}
 A_1(j)  &= \\  
&\left(1-  \exp\left(-\frac{p}{2\left(1+\left(\frac{j}{p}\right)^{\alpha-1}\right)}\left(\frac{j}{p}\right)^{2\alpha -1}\right)\right)  \frac{\P\left(v \geq   j/p+\lfloor(j/p)^\alpha\rfloor +1 \right)}{\P(v\geq j/p)},\\
A_2(j) & =   \frac{\P\left(v \geq   j/p-\lfloor(j/p)^\alpha\rfloor \right)}{\P(v\geq j/p)},
\end{align*}
 and where $B(j)$ is a positive quantity such that
 $$
 B(j) \leq  e^{-\frac{p}{2(1-p)}\left(\frac{j}{p}\right)^{2\alpha-1}}  \frac{\P(v\geq j)}{\P(v \geq j/p)}        .
 $$
 \end{theorem}

\begin{proof}
We have
$$
\P(\tilde{v}\geq j)=\P\left( \sum_{\ell=1}^v B_\ell \geq j\right) =  T_1+T_2,
$$
where
 \begin{eqnarray*}
 T_1 & =&  \P\left( \sum_{\ell = 1}^v B_\ell \geq j, j\leq v \leq j/p-\lfloor (j/p)^\alpha\rfloor-1     \right) ,\\
 T_2 &=& \P\left( \sum_{\ell = 1}^v B_\ell \geq j,    j/p-\lfloor (j/p)^\alpha\rfloor \leq v  \right) .
 \end{eqnarray*} 

 Let us first recall the following inequality for the sum of independent Bernoulli random variables $B_\ell$, $\ell \geq 1$ \cite{Siegel}: for $x \in [0,1-p]$
 \begin{equation}
 \P\left(\sum_{\ell =1}^n B_\ell -np \geq nx\right) \leq e^{-\frac{nx^2}{A(x)}},
 \end{equation}
 where 
 \begin{equation}
 A(x) = 2p(1-p)+\frac{2}{3}x(1-2p)-\frac{2}{9}x^2.
 \end{equation}
 It follows that for $j \leq v \leq j/p$
 $$
 \P\left(\sum_{\ell = 1}^v B_\ell \geq j\right) \leq e^{-\frac{(j-pv)^2}{vA\left(\frac{j}{v}-p\right)}}.
 $$
 It is easily checked that the function $v \to vA\left(\frac{j}{v}-p\right)$ is increasing in the interval $[j,j/p]$ and that for all $v \in [j,j/p]$ 
$$
vA\left(\frac{j}{v}-p\right) \leq 2j(1-p).
$$
 Hence, for $v \in [j,j/p]$
 $$
 \P\left(\sum_{\ell = 1}^v B_\ell \geq j\right) \leq e^{-\frac{(j-pv)^2}{2j(1-p)}}
 $$
 and for $v \in [j,j/p-\lfloor(j/p)^\alpha\rfloor -1]$
 $$
 \P\left(\sum_{\ell = 1}^v B_\ell \geq j\right) \leq e^{-\frac{p}{2(1-p)}\left(\frac{j}{p}\right)^{2\alpha-1}}.
 $$
 This implies that
 \begin{eqnarray*}
 T_1 &\leq&   \P\left(\sum_{\ell =1}^v B_\ell \geq j,  j \leq v \leq  j/p-\lfloor(j/p)^\alpha\rfloor -1 \right) \\ 
 &\leq &\P\left(\sum_{\ell =1}^{ j/p-\lfloor(j/p)^\alpha\rfloor -1    } B_\ell \geq j  \right)\P(v \geq j)\\
 &=&  e^{-\frac{p}{2(1-p)}\left(\frac{j}{p}\right)^{2\alpha-1}}\P(v\geq j).
 \end{eqnarray*}

 For the term $T_2$, we first note that
 $$
 T_2 \leq \P\left(v \geq   j/p-\lfloor(j/p)^\alpha\rfloor \right).
 $$
 Then, we clearly have
$$
T_2 \geq \P\left( \sum_{\ell = 1}^v B_\ell \geq j,    j/p+\lfloor (j/p)^\alpha\rfloor +1 \leq v  \right)
$$
and then 
 $$
 \frac{T_2}{\P(v \geq j/p)} \geq \P\left(\sum_{\ell =1}^{j/p  +  \lfloor(j/p)^\alpha\rfloor+1} B_\ell > j\right) \frac{\P(v\geq j/p+\lfloor (j/p)^\alpha\rfloor +1)}{\P(v \geq j/p)}.
 $$
 Chernoff bound implies for $v = j/p +  \lfloor(j/p)^\alpha\rfloor+1     $
 \begin{eqnarray*}
 \P\left(\sum_{\ell = 1}^{v}  B_\ell \leq j\right)  &\leq&\exp\left(-\frac{(pv-j)^2}{2pv}\right)  \\ 
&\leq&   \exp\left(-\frac{p}{2\left(1+\left(\frac{j}{p}\right)^{\alpha-1}\right)}\left(\frac{j}{p}\right)^{2\alpha -1}\right).
 \end{eqnarray*}
 It follows that
\begin{multline*}
 \frac{T_2}{\P(v \geq j/p)} \geq  \\ \left(1-  \exp\left(-\frac{p}{2\left(1+\left(\frac{j}{p}\right)^{\alpha-1}\right)}\left(\frac{j}{p}\right)^{2\alpha -1}\right)\right)  \frac{\P(v\geq j/p+\lfloor(j/p)^\alpha\rfloor+1)}{\P(v \geq j/p)}  .
 \end{multline*}
 and the proof follows.
 \end{proof}

The above result  can be applied to specific  distributions for $v$, namely Pareto and Weibull distributions, in order to show that the tails of the probability distribution functions of $\tilde{v}$ and $p v$ are the same. This is the analog of Proposition~\ref{taildist} for the probabilistic model.

\begin{corollary}
If $v$ has either

\noindent (1) a Pareto tail distribution with parameter $a>1$ such that for $x\geq
    0$, $\P(v\geq x)=L(x)x^{-a}$ where
    $L$ is a slowly varying function, i.e.,  for each $t>0$,
    \[
\lim_{x\rightarrow +\infty} \frac{L(tx)}{L(x)}=1;
\]
\noindent or

\noindent (2) a Weibull tail distribution with $\beta \in ]0,1/2[$ such that  for $x\geq
    0$, $\P(v\geq x)=L(x) e^{-\delta x^{\beta}}$ for some $\delta>0$
    and $L$ a slowly varying function 

then
\[
\lim_{j\rightarrow +\infty} \left| \frac{\P(\tilde{v}\geq j)}{\P(v\geq
  j/p)}-1\right|=0.
\]
\end{corollary}

\begin{proof}
For (1),
\[
\frac{\P(v\geq j)}{\P(v\geq
  j/p)}=\frac{L(j)}{L(j/p)} \frac{j^{-a}}{(j/p)^{-a}}=
  \frac{L(j)}{L(j/p)}p^a \cvj p^{-a}
\]
and
\[
\frac{\P(v\geq j/p+\epsilon (j/p)^{\alpha})}{\P(v\geq j/p)}=\frac{L(( j/p)(1+\epsilon (j/p)^{\alpha-1}))}{L(j/p)}(1+\epsilon(j/p)^{\alpha-1})^{-a}
\]
which tends to 1 when $j$ tends to $+\infty$. This implies that the quantities $A_1(j)$ and $A_2(j)$ appearing in Theorem~\ref{distance2} tends to 1 and $B(j)$ tends to 0 when $j \to \infty$. 

For (2),
\[
\frac{\P(v\geq j)}{\P(v\geq
  j/p)}=\frac{L(j)}{L(j/p)} e^{-\delta j^{\beta}(1-p^{-\beta})}
  \cvj 0
\]
and it is straightforward that
\begin{align*}
\frac{\P(v\geq j/p+\epsilon (j/p)^{\alpha})}{\P(v\geq j/p)}&=\frac{L(
    j/p(1+\epsilon (j/p)^{\alpha-1)})}{L(
    j/p)}e^{-\delta(j/p+\epsilon (j/p)^{\alpha})^{\beta}+\delta (j/p)^{\beta}}\\
&=\frac{L(
    j/p(1+\epsilon (j/p)^{\alpha-1)})}{L(
    j/p)}e^{-\delta \beta \epsilon (j/p)^{\alpha+\beta-1}(1+o(1))}
\end{align*}
which tends to $1$ if $\alpha+\beta<1$. Let $\beta \in ]0,1[$. It is
  sufficient to find $\alpha \in ]1/2,1[$ such that
  $\alpha+\beta<1$. Necessarily $1-\beta>\alpha>1/2$ thus $\beta<1/2$
  and for such a $\beta$,  such an $\alpha$ exists. 
\end{proof}

\section{Concluding remarks on sampling and parameter inference}
\label{conclusion}
We have established in this paper convergence results for the distribution of the number
of balls with a given color under the assumption that there is a large number of colors in
the urn, that the number of balls with a given color has a heavy tailed distribution
independent of the color, and that only a small fraction $p$ of the total number of ball is
sampled. We have considered two ball sampling  rules. The first one states that the
probability of drawing a ball with a given color depends upon the relative contribution of
the  color to the total number of balls and that a drawn ball is immediately replaced into
the urn. With the second rule, each ball is selected with probability $p$ independently of
the others. The two rules do not give the same results, even if they coincide when $p\to
0$ (see \cite{Chabchoub:01} for details). 

From a practical point of view, we have shown that it is possible to identify the original
distribution of the number of balls with a given color by using the tail of the 
distribution of the number of balls with a a given color drawn from the urn. A stronger
result holds for Pareto  when the number of colors is
very large (see Proposition~\ref{propalpareto}). This result is robust in practice because
it does not rely on the asymptotics of the tail distribution (in
Proposition~\ref{propalpareto} assertions hold for all $j >a$). 

The  determination of the original number of balls per color is valid when the number of
balls follows a unique distribution of Pareto or Weibull type. This could be used in the
context of packet sampling in the Internet. In practice, however, the number of packets in
flows is in general not described by a unique ``nice'' distribution, but can only be
locally approximated by a series of Pareto distributions (see \cite{Info08} for a
discussion). More sophisticated techniques are then necessary to get the original
statistics of flows.

\providecommand{\bysame}{\leavevmode\hbox to3em{\hrulefill}\thinspace}
\providecommand{\MR}{\relax\ifhmode\unskip\space\fi MR }
\providecommand{\MRhref}[2]{%
  \href{http://www.ams.org/mathscinet-getitem?mr=#1}{#2}
}
\providecommand{\href}[2]{#2}

\end{document}